\documentclass[paper=a4,fontsize=11pt,twoside]{scrartcl}  
\title{Interpretable hypothesis tests}

\author{Victor Coscrato, 
Lu\'{i}s G. Esteves, 
Rafael Izbicki, 
Rafael B. Stern}

\author{Victor Coscrato \\ Federal University of S\~ao Carlos \\ \and  
	Lu\'is Gustavo Esteves							\\
	University of S\~ao Paulo					\\ 
	\and   Rafael Izbicki							\\
	Federal University of S\~ao Carlos \\
	\and Rafael Bassi Stern						\\
	Federal University of S\~ao Carlos \\
}

\usepackage[T1]{fontenc}
\usepackage{nameref}
\usepackage[colorlinks=true, urlcolor=blue, citecolor=blue, linkcolor=blue]{hyperref}
\usepackage{amsfonts, amsmath, amssymb, amsthm, thmtools}
\usepackage{mathrsfs}
\usepackage{cleveref}

\declaretheorem[name=Theorem, refname={Theorem, Theorems}, Refname={Theorem, Theorems}, parent=section]{theorem}

\declaretheorem[name=Definition, refname={Definition, Definitions}, Refname={Definition, Definitions}, sibling=theorem, style=definition]{definition}
\declaretheorem[name=Example, refname={Example, Examples}, Refname={Example, Examples}, sibling=theorem, style=definition]{example}

\crefname{section}{section}{sections}
\Crefname{section}{Section}{Sections}
\crefname{table}{table}{tables}
\Crefname{table}{Table}{Tables}

\usepackage{color, enumitem, float, fourier, layout, multirow, setspace, verbatim}
\setlist[enumerate]{leftmargin=*}
\usepackage[colorinlistoftodos]{todonotes}
\usepackage{listings}
\usepackage{caption}
\usepackage{epigraph}

\usepackage{natbib}
\usepackage{graphicx}
\graphicspath{{figures/}{../figures/}}

\def\half{\frac{1}{2}}

\def\I{{\mathbb I}}

\def\x{{\vec{x}}}

\def\X{{\vec{X}}}

\def\P{{\mathbb P}}

\def\Re{{\mathbb R}}

\def\sD{\mathcal{D}}

\def\sX{\mathcal{X}}

\def\t0{{\theta_0}}
\def\ts{{\theta^*}}
\def\tso{{\theta^*_1}}
\def\tst{{\theta^*_2}}

\definecolor{darkolivegreen}{rgb}{0.33, 0.49, 0.18}

\renewcommand{\vec}[1]{\mathbf{#1}}

\usepackage{subfig}
\usepackage{setspace}
\def\to{{\theta_1}}
\def\ts{{\theta^*}}
\def\tz{{\theta_0}}
\def\Tz{{\Theta_0}}
\def\Tzc{{\Theta_0^c}}

\def\Z{{\vec{Z}}}
\def\z{{\vec{z}}}

\def\sD{\mathcal{D}}

\def\sZ{\mathcal{Z}}
\def\sX{\mathcal{X}}

\def\sD{\mathcal{D}}

\def\sZ{\mathcal{Z}}
\def\sX{\mathcal{X}}

\newcommand{\gfbstfig}{
 \begin{tikzpicture}[thick,scale=0.75, every node/.style={scale=0.9}]
  \draw (-6,2) circle (1.5);
  \draw [fill=lightgray] (-6,2) circle (0.8);
  \draw (-8,0) -- (-4,0) -- (-4,4) -- (-8,4) -- (-8,0);;
  \node at (-6,2) {{\large$R(\x)$}}; 
  \node at (-6,0.8) {{\large$H_0$}};
  \node at (-7.5,.50) {{\large$H_0^c$}};
  \node at (-6,4.25) {{\large $\phi(\x)=0$}};
			
  \draw (-1.4,2.4) circle (1.1);
  \draw [fill=lightgray] (0.2,0.9) circle (0.7);
  \draw (-3,0) -- (1,0) -- (1,4) -- (-3,4) -- (-3,0);;
  \node at (0.2,0.9) {{\large$R(\x)$}}; 
  \node at (-1.4,2.4) {{\large$H_0$}};
  \node at (-2.5,.50) {{\large$H_0^c$}}; 
  \node at (-1,4.25) {{\large $\phi(\x)=1$}};	
			
  \draw [fill=lightgray] (5.2,1.9) circle (0.7); 
  \draw (4.4,2.4) circle (1.1);
  \draw (2,0) -- (6,0) -- (6,4) -- (2,4) -- (2,0);;
  \node at (5.2,1.9) {{\large$R(\x)$}}; 
  \node at (4.2,2.5) {{\large$H_0$}};
  \node at (2.5,.50) {{\large$H_0^c$}};
  \node at (4,4.25) {{\large $\phi(\x)=1/2	$}};	
 \end{tikzpicture} \\ 
}

\usepackage{algorithm}
\usepackage{algpseudocode}
\renewcommand{\algorithmicrequire}{\textbf{\small Input:}}
\renewcommand{\algorithmicensure}{\textbf{\small Output:}}

\usepackage{xpatch}
\xpatchcmd{\algorithmic}{\setcounter}{\algorithmicfont\setcounter}{}{}
\providecommand{\algorithmicfont}{}

\begin{document}
\maketitle

\begin{abstract}
 \textbf{Abstract}: 
 Although hypothesis tests play
 a prominent role in Science,
 their interpretation can be challenging.
 Three issues are
 (i) the  difficulty in
 making an assertive decision 
 based on the output of an hypothesis test,
 (ii) the logical contradictions that occur in
 multiple hypothesis testing, and (iii)
 the possible lack of practical importance
 when rejecting a precise hypothesis.
 These issues can be addressed through
 the use of agnostic tests and
 pragmatic hypotheses.
\end{abstract}


\textbf{ Running title: Interpretable  hypothesis tests}

\section{The elements of
interpretable hypothesis tests}
 
Although hypothesis tests play
a prominent role in Science,
its importance has been downplayed recently
\citep{Diggle2011,Wasserman2013,Trafimow2018,Pike2019}.
A major reason why hypothesis tests have been
criticized is that they can be difficult to interpret and
can even lead to misleading conclusions
\citep{Greenland2016,Kadane2016,Wasserstein2019}.
At least three issues contribute to these challenges:

\begin{enumerate}[label=$\bullet$ \textbf{Issue (\roman*)}:, wide, labelwidth=!, labelindent=0pt]
 \item \textbf{At least one of the outcomes of
 a standard hypothesis test is hard to interpret}.
 While a standard hypothesis test can either
 reject or not reject the null hypothesis,
 the data can lead to at least three types of credal state:
 favor the null hypothesis, favor the alternate hypothesis,
 or remain undecided. Therefore,
 the test will assign to the same output at least
 two different credal states.
 For example, standard frequentist hypothesis tests usually
 do not control type II error probabilities. 
 As a result, the non-rejection of 
 the null hypothesis, $H_0$, can either be
 due to lack of evidence to reject $H_0$ or
 due to evidence in favor of $H_0$ \citep{Fisher1959}.
 For instance, consider that $X \sim N(\theta, 1)$ and
 $H_0: \theta \geq 0$. In a z-test, 
 $H_0$ is not rejected no matter whether
 $X$ is close to $0$ or $X$ is very large.
 However, while in the former case
 there is little evidence in favor or against $H_0$, 
 in the latter there is evidence in favor of $H_0$.
 \citet{Edwards1963} approaches this junction by
 stating that ``if the null hypothesis is not rejected, 
 it remains in a kind of limbo of suspended disbelief''.
 As a result, although in practice
 the non-rejection of $H_0$ is
 often taken as evidence in favor of $H_0$,
 this conclusion is not warranted by the test.

 \item \textbf{For standard Bayesian and classical tests,
 multiple hypothesis testing can lead to
 logically incoherent conclusions} \citep{Izbicki2015,Fossaluza2017}.
 For example, for a given dataset,
 a test might reject $H_0: \theta \leq 0$ and
 not reject $H_0: \theta = 0$, although
 the latter implies the former. Similarly, a test might
 reject both $H_0$ and $H_0^*$ but not reject
 $H_0 \cup H_0^*$. These logical contradictions are
 hard to interpret and explain.

 \item \textbf{When a precise hypothesis is rejected,
 this outcome does not mean that
 the rejection is relevant from 
 a practical perspective}.
 For instance, consider that populations $1$ and $2$
 are composed of, respectively, healthy and sick persons.
 Furthermore, for each person, 
 one can observe a clinical variable, $X$,
 such as the patient's average blood glucose level.
 Assume that, if $X_i$ is
 a person from population $i$, then
 $X_i \sim N(\theta_i,1)$.
 Rejecting that the populations are the same,
 $H_0: \theta_1 = \theta_2$, does not
 imply that $X$ can be used to determine whether
 a person is healthy or sick.
 For instance, when the sample size is large,
 one might reject $H_0$ although
 $|\theta_1-\theta_2|$ is not equal but close to $0$.
 In this case $X$ cannot effectively be used to
 determine whether a person is healthy or sick.
 This result can lead to counter-intuitive policies,
 such as considering an experiment to 
 be inadequate from a statistical perspective because
 its sample size is too large \citep{Faber2014}.
 Although solutions to this problem have been proposed,
 such as considering effect sizes \citep{Cohen1992},
 they also increase the difficulty in
 interpreting hypothesis tests.
\end{enumerate}

This paper shows that 
the above challenges in interpretation are
avoided by making simple changes to
the practice of hypothesis tests.
These changes have two key components:

\begin{enumerate}[label=(\alph*), wide, labelwidth=!, labelindent=0pt]
 \item \emph{Agnostic hypothesis tests}
 \citep{Neyman1976,Berg2004},
 which parallel agnostic classifiers 
 \citep{Lei2014,Jeske2017,Jeske2017b,Sadinle2017} and
 allow three possible results:
 reject $H_0$, accept $H_0$, or remain agnostic.
 The last option permits a test to control both
 the type I and type II errors \citep{Coscrato2018},
 avoiding the ``limbo of suspended disbelief'' described in issue (i)
 that follows from the
 non-rejection of $H_0$ in standard tests.
 This occurs because the option of
 remaining agnostic allows
 the test to explicitly indicate when data does not
 provide substantial evidence either
 in favor or against $H_0$.
 Although agnostic tests introduce a type III error,
 which occurs whenever the test remains agnostic,
 this error is qualitatively different from
 the errors of type I and II.
 While it is unknown when the latter errors occur,
 the error of type III is known.
 Hence, the user of an agnostic test
 can control unknown errors and
 either acknowledge errors of type III or
 correct them by, for example,
 collecting more data.
 Besides these benefits,
 \citet{Esteves2016,Stern2017} show that,
 as opposed to standard tests,
 agnostic tests can guarantee
 logically coherent conclusions in
 multiple hypothesis testing, solving
 issue (ii).

 \item \emph{Pragmatic hypothesis}, which
 substitute precise hypotheses
 whenever the goal of the test is
 to determine variables with
 good predictive capabilities.
 For instance, let
 $X_1$ and $X_2$  be
 average blood glucose levels of
 healthy and sick persons and
 $X_i \sim N(\theta_i, \sigma^2)$.
 If one wishes to discover whether
 the blood glucose level is
 useful in determining whether
 an individual is healthy or sick, then
 the rejection of an hypothesis such as
 $Pg(H_0): |\theta_1-\theta_2| > k\sigma$ 
 \citep{Chow2016}  is more informative than
 the rejection of $H_0: \theta_1 \neq \theta_2$, which solves issue (iii).
 Similar ideas for augmenting the null hypothesis 
 have previously been proposed in
 \citep{Berger2013,Degroot2012}.   
\end{enumerate}

The following sections define, illustrate and
describe how to build agnostic tests and
pragmatic hypotheses. This task requires additional notation.
Specifically, the hypotheses that are considered are
propositions about a parameter,
$\theta \in \Theta \subseteq \Re^d$.
A null hypothesis is a proposition of
the form $H_0: \theta \in \Tz$, where
$\Tz \subseteq \Theta$ and
the alternative hypothesis, $H_1$,
is $H_1: \theta \in \Tzc$.
Whenever there is no ambiguity,
$\Tz$ is used instead of $H_0$.
In order to test $H_0$,
data, $\X \in \sX$, is used.
Finally, the data follows
a distribution given by
$\P_{\tz}$ when $\theta=\tz \in \Theta$.

\section{Agnostic hypothesis tests}
\label{sec:agnostic}

\setlength{\epigraphwidth}{11.5cm}
\epigraph{``The phrase `do not reject H' is
longish and cumbersome \ldots (This action) should be
distinguished from (the ones in) a 
`three-decision problem' (in which the) actions are:
(a) accept H, (b) reject H, and
(c) remain in doubt.''}{\citet{Neyman1976}}

A challenge in the interpretation of
standard hypothesis tests is that
they must always conclude one out of
two possibilities. Although 
only two conclusions are available, data lead to
at least three credal states:
strongly disfavor $H_0$, strongly favor $H_0$ or
not strongly favor or disfavor $H_0$.
Standard tests usually assign
the latter two states to
the "non-rejection" of $H_0$.
As a result, standard tests
assign datasets which are
qualitatively different to
the same conclusion.

This challenge is addressed by
agnostic tests, which can
accept $H_0$ (0), reject $H_0$ (1) or
remain undecided $\left(\half\right)$.
The set of possible outcomes of such a test is
denoted by $\mathcal{D} = \left\{0,\half,1\right\}$.

\begin{definition}
 \label{def:agnostic}
 An agnostic test is a function,
 $\phi: \sX \rightarrow \sD$.
\end{definition}

\begin{definition}
 \label{def:std}
 An agnostic test, $\phi$, is a
 standard test if Im$[\phi]=\{0,1\}$.
\end{definition}

An agnostic z-test is presented in
\cref{ex:gaussian}.

\begin{example}
 \label{ex:gaussian}
 Let $X \sim N(\theta,1)$. 
 The usual $\alpha$-level z-test for 
 $H_0:\theta_0 \geq 0$ determines:
 \begin{equation*}
  \begin{cases}
    \text{reject $H_0$}        & \text{, if } X \leq \Phi^{-1}(\alpha) \\
    \text{don't  reject $H_0$} & \text{, if } X > \Phi^{-1}(\alpha)
  \end{cases}
 \end{equation*}
 For $\alpha=0.05$, $\Phi^{-1}(\alpha)$ is approximately $-1.64$.
 Therefore, no matter whether $x=-0.5$ or $x=10^{100}$,
 $H_0$ is not rejected. That is, 
 no assertive decision about $H_0$ is obtained in either case
 although $x=10^{100}$ favors $H_0$ and
 $x= -0.5$ does not.
 
 Alternatively, one can test $H_0$ with
 an agnostic test as follows:
 \begin{equation}
  \label{eq:agnostic_1}
  \begin{cases}
   \text{reject $H_0$} & \text{, if } X < -\Phi^{-1}(0.5\alpha) \\
   \text{accept $H_0$} & \text{, if } X > \Phi^{-1}(1-0.5\alpha) \\
   \text{remain agnostic} & \text{, otherwise}
  \end{cases}
 \end{equation}
 For $\alpha=0.05$, $\Phi^{-1}(1-0.5\alpha)$ is approximately $1.96$.
 Therefore, while $x=-0.5$ leads to an agnostic decision,
 $x=10^{100}$ leads to the assertive decision of accepting $H_0$.
 Contrary to the standard z-test,
 the agnostic test distinguishes these
 qualitatively different types of data.
\end{example}

An agnostic test can have $3$ types of errors.
The type I and type II errors of
agnostic tests are defined in the same way
as those of standard tests. That is, 
a type I error occurs when the test 
rejects $H_0$ and $H_0$ is true.
Similarly, a type II error occurs when
the test accepts $H_0$ and $H_0$ is false.
A type III error occurs whenever the
test remains agnostic. That is,
contrary to type I and type II errors,
one knows when type III errors occur.
Given this asymmetry, one might design
either frequentist or Bayesian tests that
control the errors of type I and II,
as presented in \cref{defn:level} and
\cref{defn:blevel}.

\begin{definition}
 \label{defn:level}
 An agnostic test, $\phi$, has
 $(\alpha,\beta)$-level if the
 test's probabilities of committing errors
 of type I and II are controlled
 by, respectively, $\alpha$ and $\beta$.
 That is,
 \begin{align*}
  \alpha_{\phi} &:= 
  \sup_{\tz \in H_0}
  \P_{\tz}(\phi = 1) \leq \alpha \\
  \beta_{\phi} &:= 
  \sup_{\to \in H_1}
  \P_{\to}(\phi = 0) \leq \beta
 \end{align*}
 Similarly, $\phi$, has size
 $(\alpha,\beta)$ if
 $\alpha_{\phi} = \alpha$ and
 $\beta_{\phi} = \beta$.
\end{definition}

\begin{definition}
 \label{defn:blevel}
 An agnostic test, $\phi$, has
 false conclusion probability of $\gamma$
 according to a prior distribution over $\theta$,
 $f_{\theta}: \Theta \rightarrow \Re^{+}$, if
 \begin{align*}
  \gamma_{\phi} &:= 
  \int_{\tz \in H_0}\P_{\tz}(\phi = 1)f_{\theta}(\tz)d\tz
  +\int_{\to \in H_1}\P_{\to}(\phi = 0)f_{\theta}(\to)d\to
  = \gamma
 \end{align*}
\end{definition}

There are several ways of controlling 
the errors above. For instance,
\citep{Esteves2016,Coscrato2018}
discuss approaches based on 
statistical decision theory.
The following subsection presents an
agnostic test that controls 
the errors above while
preserving other properties,
such as logical consistency.

\subsection{Region-based agnostic tests}
\label{sec:region}

Agnostic tests can be constructed
through region estimators.
A region estimator is a function that
assigns a subset of the parameter space to
each possible dataset.
Generally, a region estimator can
be interpreted as a set of
plausible values for $\theta$.
For instance, when $\Theta = \Re$,
confidence and credible intervals are
region estimators.

\begin{definition}
 \label{defn:region}
 A region estimator, $R$, is a function
 from $\sX$ such that $R(\x) \subseteq \Theta$.
\end{definition}

It is possible to completely specify
an agnostic test by means of a region estimator.
Based on the idea that the region estimator
indicates the plausible values for $\theta$,
there are three cases to consider.
If all plausible values lie in $H_0$, then
there is strong evidence in favor of $H_0$ and
$H_0$ is accepted. Also,
if all plausible values lie outside of $H_0$, then
there is strong evidence against $H_0$ and
$H_0$ is rejected. Finally,
if there are plausible values both in and
outside of $H_0$, then $H_0$ remains undecided.
\Cref{defn:agnostic-region} formalizes this description.

\begin{definition}
 \label{defn:agnostic-region}
 The agnostic test based on $R$ for
 testing $H_0$, $\phi$,
 illustrated in \Cref{fig::region}, is
 \begin{align*}
  \phi(\x) &=
  \begin{cases}
   0 & \text{, if } R(\x) \subseteq H_0 \\
   1 & \text{, if } R(\x) \subseteq H_0^c \\
   \half & \text{, otherwise.} \\
  \end{cases}
 \end{align*}
 \begin{figure}
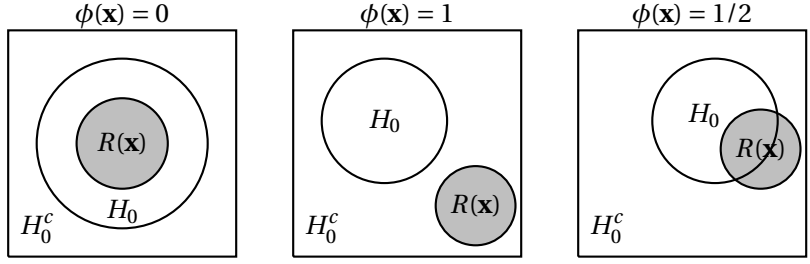

  \center
  \gfbstfig
  \mbox{} \vspace{-2mm} \mbox{} \\ 
  \caption{$\phi(\x)$ is
  an agnostic test based on 
  the region estimator, $R(\x)$, 
  for testing $H_0$.}
  \label{fig::region}
 \end{figure}
\end{definition}

If an agnostic tests is based on a region estimator
that is a (frequentist) confidence set, then 
the size of the test is controlled,
as described in \cref{thm:region-level}.

\begin{theorem}[\citet{Coscrato2018}]
 \label{thm:region-level}
 If $R(\x)$ is a region estimator for
 $\theta$ with confidence $1-\alpha$ and
 $\phi$ is an agnostic test for $H_0$
 based on $R$ (\cref{defn:agnostic-region}), then
 $\phi$ is a $(\alpha,\alpha)$-level test. Also,
 for every prior distribution over $\theta$, $f_\theta$,
 $\gamma_{\phi} \leq \alpha$ according to $f_\theta$.
\end{theorem}

Similarly, if an agnostic test is based on 
a (Bayesian) credible set instead of a confidence set, then
it controls the false conclusion probability,
as described in \cref{thm:region-fcp}.

\begin{theorem}
 \label{thm:region-fcp}
 If $R(\x)$ is a region estimator for
 $\theta$ with credibility $1-\gamma$
 according to $f_{\theta}$ and
 $\phi$ is an agnostic test for $H_0$
 based on $R$ (\cref{defn:agnostic-region}), then
 $\gamma_{\phi} \leq \gamma$
 according to $f_{\theta}$.
\end{theorem}

\Cref{ex:gaussian} below shows how
\cref{thm:region-level,thm:region-fcp}
can construct agnostic z-tests.

\begin{example}[Continuation of \cref{ex:gaussian}]
 \label{ex:gaussian_2}
 The agnostic test in \cref{ex:gaussian} can be
 obtained from \cref{thm:region-level} by
 using the usual ($1-\alpha$)-level
 confidence interval given by
 $CI = [X-\Phi^{-1}(0.5\alpha),X-\Phi^{-1}(1-0.5\alpha)]$.
 The obtained test has $(\alpha,\alpha)$-level.
 
 One can also test $H_0$ by applying \cref{thm:region-fcp}.
 If $\theta \sim N(0,1)$, then
 a typical $1-\gamma$ credible interval for $\theta$ is
 $CI = \left[0.5X - \frac{\Phi^{-1}(0.5\gamma)}{\sqrt{2}} , 
 0.5X - \frac{\Phi^{-1}(1-0.5\gamma)}{\sqrt{2}}\right]$.
 The test based on this interval controls
 the false conclusion probability by $\gamma$ and
 behaves similarly to the test in \cref{eq:agnostic_1}.
\end{example}

An example of a test based on
a credible sets is the
Generalized Full Bayesian Significance Test (GFBST) 
\citep{Stern2017}, which is obtained
when the credibility set is
a highest posterior density set.
\Cref{thm:region-fcp} shows that
the GFBST controls the false conclusion probability.

Another useful property of
tests based on region estimators,
such as the ones obtained from
\cref{thm:region-level,thm:region-fcp},
is that they are the only logically coherent tests
\citep{Esteves2016}. That is, if
the same region estimator is used for
simultaneously testing several hypothesis, then
there will be no logical contradiction between
the conclusions of the tests. For instance,
a standard t-test can reject $\theta \leq 0$ and
not reject $\theta = 0$. Since $\theta \leq 0$ is
implied by $\theta = 0$, such a result is
a logical contradiction.
If an agnostic test based on
a region estimator rejects $\theta \leq 0$, then
it also rejects $\theta = 0$.

It follows from the logical coherence of
agnostic tests based on region estimators that
they are consistent with
propositional logic \citep{Stern2017}[Lemma 6.1].
For example, consider a class of
agnostic tests, $\phi$, based on
the same region estimator,
two hypotheses, $H_1$ and $H_2$, and
the logical proposition
$H^* = P(H_1, H_2) = (H_1 \cap H_2)^c$.
Since $\phi$ is consistent with propositional logic,
the outcome of testing $H^*$ with $\phi$ can be
determined by the outcome of testing
$H_1$ and $H_2$ with $\phi$. For example,
if $\phi$ rejects either $H_1$ or $H_2$,
then it accepts $H^*$,
if $\phi$ accepts $H_1$ and $H_2$, then
it rejects $H^*$, and
otherwise it remains agnostic about $H^*$.
The consistency with propositional logic not
only makes the test easier to interpret but
also makes simultaneous hypothesis testing
easier to implement. The latter occurs since
calculating the truth-value of a proposition is
generally less expensive computationally than 
a direct calculation of the test.



Despite the advantages of
region-based agnostic tests over
standard test, the former
usually do not accept precise hypotheses.
For instance, if $H_0: \theta = 0$, then
whenever a confidence interval contains
more than a single point, it follows from
\cref{defn:agnostic-region} that
$H_0$ is not accepted.
The following section argues that
this result is justifiable.
From a practical perspective,
whenever one wishes to
be able to accept the null hypothesis,
this hypothesis can be well represented by
a pragmatic hypothesis.

\section{Pragmatic hypotheses}
\label{sec:pragmatic}

\setlength{\epigraphwidth}{11.5cm}
\epigraph{``The null hypothesis is really a hazily defined
small region rather than a point.''}{\citet{Edwards1963}}

When the null hypothesis is stated as an equality,
it is often reasonable to enlarge it to
a set of values which are close to
satisfying the equality from
a practical perspective.
Such an enlarged hypothesis is
called a pragmatic hypothesis.
Although in some situations
the pragmatic hypothesis might be
derived from expert knowledge, 
this solution might not always be available.
This section presents a method for
deriving pragmatic hypotheses which
closely resembles the ones in \citet{Esteves2019}.

We assume that
the researcher is interested in predicting
a future experiment, $\Z \in \sZ$, which 
is distributed according to
a density, $f_{\Z}(\z|\theta)$.
This future experiment can be different from
$\X$, which is used to test the hypothesis.
Specifically, the hypothesis is tested in the present using
$\X$ so that accurate predictions about $\Z$ can
be made in the future.

For a given future experiment $\Z$,
one can determine which values of $\theta$ make
$\Z$ behave similarly.
A \emph{predictive dissimilarity function}, 
$d_{\Z}: \Theta^2 \rightarrow \Re^{+}$ is a function
$d_{\Z}(\t0, \ts)$  that measures how much $\Z$
behaves differently under $\t0$ and under $\ts$.
We focus on the \emph{classification dissimilarity}:
\begin{align}
 \label{defn:dissim}
 d_{\Z}(\t0,\ts)
 &= 0.5\left(\P_{\t0}\left(\frac{f(\Z|\t0)}{f(\Z|\ts)} > 1\right)
 + \P_{\ts}\left(\frac{f(\Z|\ts)}{f(\Z|\t0)} > 1 \right) \right)
\end{align}

The classification dissimilarity can be interpreted
using the Neyman-Pearson lemma \citep{Neyman1933} as follows.
Consider that $\Z$ was generated with
equal probability  either from $\t0$ or from $\ts$.
After observing $\Z$ in such a situation,
the classification dissimilarity is
the highest achievable probability of correctly identifying
which $\theta$ generated $\Z$.

Once a predictive dissimilarity function is chosen,
the pragmatic hypothesis associated to  $H_0$,
$Pg(H_0)$, is defined as
the set of parameter values whose
dissimilarity to $H_0$ is
at most $1-\epsilon$:
\begin{align*}
 Pg(H_0) &= \bigcup_{\tz \in H_0}
 \left\{\ts \in \Theta: d_{\Z}(\tz, \ts) < 1-\epsilon\right\}.
\end{align*}
If $\ts \notin Pg(H_0)$, then
$\Z$ can be used to discriminate between
$\ts$ and any given point in $H_0$ with
an accuracy of at least $1-\epsilon$.

This construction can be illustrated with
a test of equality between populations.
In this case, if a parameter value lies outside of
the pragmatic hypothesis, then
there exists a classifier based on $\Z$
with accuracy of at least $1-\epsilon$ for
determining which population generated $\Z$ 
(\cref{thm:cd_2pop}).
This procedure is applied to real data in
\cref{ex:classification}.

\begin{example}
 \label{ex:classification}
 The Cambridge Cognition Examination (CAMCOG)
 \citep{Roth1986} is a questionnaire that
 is used to measure the extent of dementia and
 assess the level of cognitive impairment.
 We use the data from \citet{Cecato2016} to
 check whether CAMCOG is able to
 distinguish three groups of patients:
 (i) control (CG),
 (ii) mild cognitive impairment (MCI), and
 (iii) Alzheimer's disease (AD).
 We assume that, if
 $Y_{i,j}$ is the score of the
 $j$-th patient in group $i$,
 then $Y_{i,j} = \mu_i + \epsilon_{i,j}$,
 where $\mu_i$ are the population averages in
 each group and $\epsilon_{i,j}$ are
 independent variables such that 
 $\epsilon_{i,j} \sim N(0,\sigma^2)$.
 
 \begin{figure}
  \centering
  \includegraphics[scale=0.6]{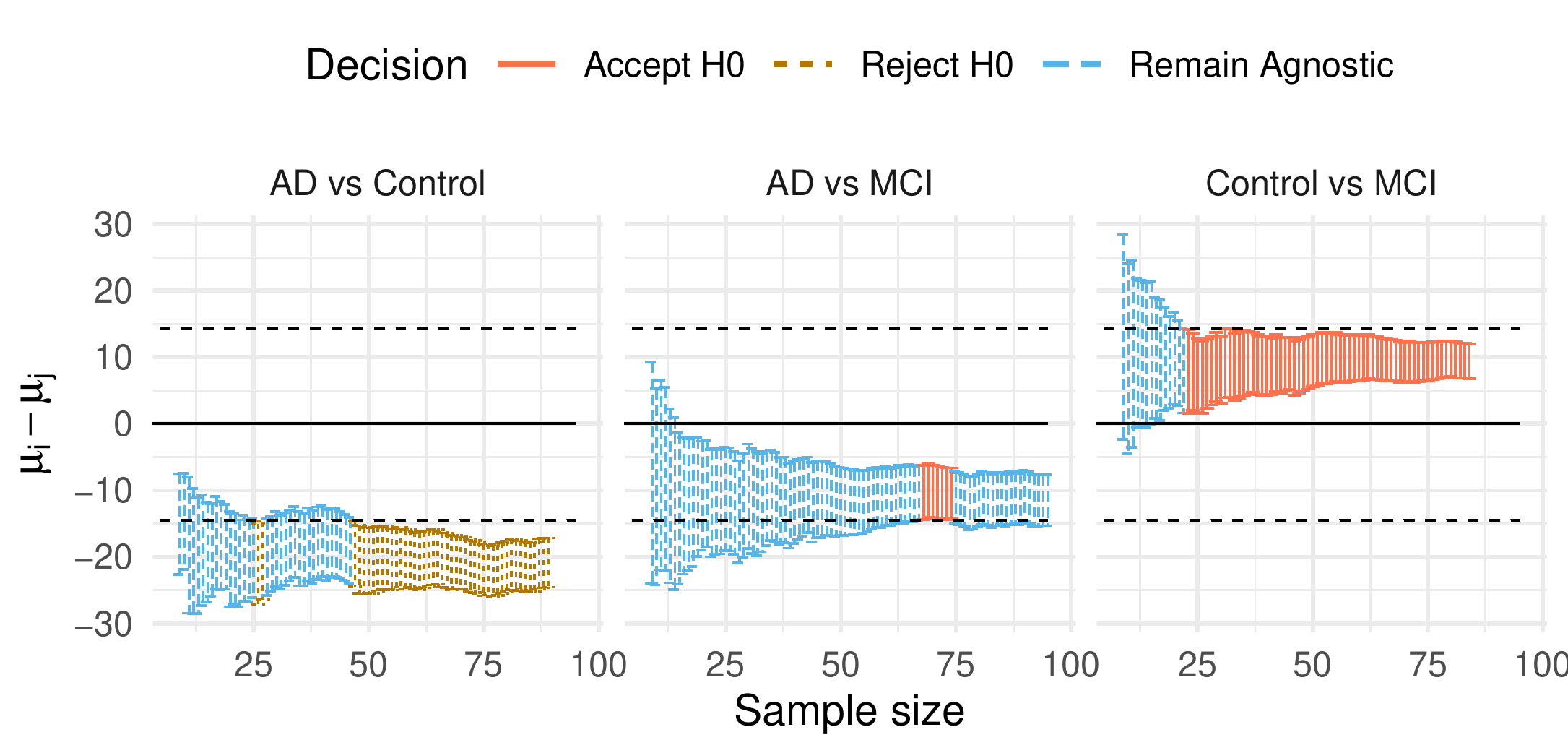}%
   \caption{Confidence intervals for
    the average difference between groups, $\mu_i-\mu_j$
    as a function of the sample size.
    The solid line indicates the precise hypothesis 
    considered in each figure, $H_0: \mu_i = \mu_j$.
    The dashed horizontal lines delimit
    the pragmatic null hypotheses that are
    induced by each precise hypothesis.}
  \label{fig:pragmatic_camcog}
 \end{figure}
 
 \Cref{fig:pragmatic_camcog} illustrates how
 the sample size affects
 region-based agnostic tests
 (\cref{defn:agnostic-region}) for
 testing the pragmatic hypotheses induced by
 $H_0: \mu_i - \mu_j = 0$.
 In each of the plots, the
 solid line indicates the
 precise hypothesis of interest,
 that is, the CAMCOG scores are
 equally distributed among
 the compared groups.
 The values of $\mu_i-\mu_j$ between
 the dashed lines compose
 the associated pragmatic hypothesis, $Pg(H_0)$.
 When the pragmatic hypothesis does not hold,
 there exists a classifier based on
 the CAMCOG score which highly discriminates
 the groups under comparison.
 For small sample sizes, the test remains agnostic
 about all of the three hypotheses.
 As the sample size increases,
 the pragmatic hypothesis associated to
 AD vs Control is rejected,
 the one for AD vs MCI is undecided and
 the one for Control vs MCI is accepted.
 That is, although it is unclear whether
 a classifier based on the CAMCOG score can
 highly discriminate between patients with AD and MCI,
 it can highly discriminate between AD and Control and
 cannot highly discriminate between Control and MCI.
 Since AD is an aggravation of MCI, these conclusions are
 compatible with qualitative knowledge.
\end{example}

In situations with several parameters,
it can be expensive to compute $Pg(H_0)$ exactly.
In these cases, it is possible to calculate
an approximation of $Pg(H_0)$. For example,
let $\theta(i)$ denote
the i-th coordinate of $\theta$ and
consider that $H_0: \theta(1) = \tz$.
In this case, an approximate pragmatic hypothesis,
$Pg^*(H_0)$, can be obtained by considering that
the remaining parameter coordinates are 
equal to a given estimate. Specifically,
\begin{enumerate}
 \item Estimate $\theta(2),\ldots,\theta(d)$ based on $\X$
 with $\hat{\theta}(2),\ldots,\hat{\theta}(d)$.
 
 \item Define $g: \Re \rightarrow \Re^d$ such that
 $g(t) = \left(t, \hat{\theta}(2),\ldots,\hat{\theta}(d)\right)$.
 
 \item Let
 $Pg^*(H_0) = \left\{\ts \in \Re: 
 d_{\Z}(g(\t0),g(\ts)) \leq \epsilon \right\}
 \times \Re^{d-1}$.
\end{enumerate}

\Cref{alg:pragmatic} in \Cref{sec:approx-pragm} shows how
such a procedure can be applied to
a generic model using Monte Carlo integration.
An implementation of this procedure in R 
is available at \url{https://github.com/vcoscrato/pragmatic}.
\Cref{ex:regression} illustrates this procedure in
the context of linear regression.

\begin{example}
 \label{ex:regression}
 An object dropped from a vertical distance $d$ from
 the ground takes $T = \sqrt{\frac{2d}{g}}$ units of
 time to reach the floor, where $g$ is
 Earth's gravitational acceleration.
 \citet[Chapter 2]{Diggle2011} describes
 a lab experiment for estimating $g$:
 a student drops an object from
 several heights and
 measures how long it takes to
 reach the ground by using a chronometer.
 Since the student has a reaction for
 activating and deactivating the chronometer,
 the data may be modeled as
 \begin{align*}
  T &= \beta_0 + \beta_1 x + \epsilon,
  & \text{where } x = \sqrt{d}, 
  \beta_1=\sqrt{2g^{-1}} \text{ and }
  \epsilon \sim N(0,1)
 \end{align*}
 
 One might be interested in testing $H_0: g = 9.8$.
 Besides $g$, the parameter space also includes
 the average reaction time of the students, $\beta_0$ and
 the imprecision in their measurements, $\sigma^2$.
 Although obtaining $Pg(H_0)$ is not intractable in this case,
 it would involve a search in a three-dimensional space.
 This procedure is simplified when calculating $Pg^*(H_0)$,
 in which $\beta_0$ and $\sigma^2$ are considered to be equal to
 their estimates that are obtained from 
 the observed sample, $\X$.
 As a result, determining $Pg^*(H_0)$ involves
 a search over a one-dimensional space only.
 
 \Cref{fig:pragmatic_g} illustrates how
 the sample size affects
 region-based agnostic tests for $Pg^*(H_0)$.
 In the left and right plots,
 the precise hypotheses are that, respectively,
 $H_0: g = 9.5$ and $H_0: g = 9.8$.
 In each plot, the values of $g$ delimited between
 the horizontal dashed lines constitute
 the pragmatic null hypotheses, $Pg^*(H_0)$.
 These pragmatic hypotheses are composed by
 values of $g$ which would induce
 predictions for future experiments in 
 a similar way as each precise hypothesis.
 These pragmatic hypotheses are
 tested with a region-based agnostic test.
 For small sample sizes,
 the test remains undecided about
 both pragmatic hypotheses.
 As the sample size increases,
 the test rejects that
 $g$ is close to $9.5$ and
 accepts that $g$ is close to $9.8$.
 The latter conclusion might seem incorrect,
 since $g \approx 9.807$. However,
 given the high imprecision in
 the experiment performed by the students,
 future experiments would behave similarly
 no matter whether $g = 9.807\ldots$ or $g = 9.8$. 
 \begin{figure}
 \centering
  \includegraphics[scale=0.6]{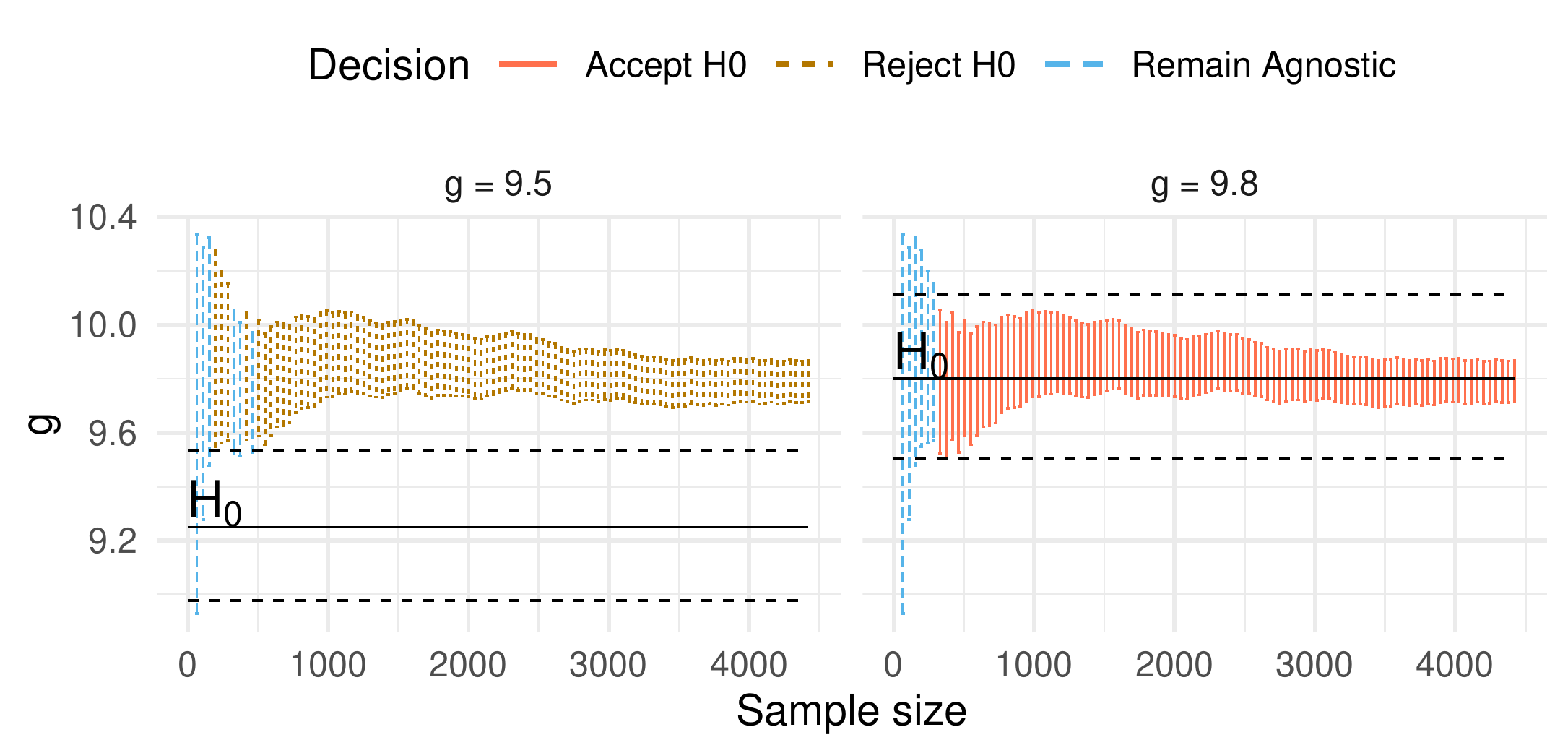}%
  \caption{Confidence intervals for
   the gravitational constant $g$ as a function of
   the sample size. The solid line indicates
   the precise hypothesis considered in each figure.
   The dashed horizontal lines delimit
   the pragmatic null hypotheses that are
   induced by each precise hypothesis.
  }
  \label{fig:pragmatic_g}
 \end{figure}
\end{example}

\section{Conclusions and future research}

Challenges in the interpretation of
standard hypothesis tests can
be addressed through
changes in statistical practice.
Agnostic hypothesis tests
lead to test outputs that
are easier to interpret and
also avoid logical contradictions in
multiple hypothesis testing.
Also, the use of pragmatic hypotheses
render that the rejection of
the null hypothesis is of practical importance.
\cref{ex:classification,ex:regression}
illustrate how these improvements admit
a simple implementation in standard models.

This paper also provides a general method for
obtaining approximate pragmatic hypotheses in
parametric statistical models.
Future research might involve obtaining
pragmatic hypothesis in nonparametric models and
a decision-theoretic approach to agnostic tests.

\section*{Acknowledgments}

This study was also financed in part by
the Coordena\c{c}\~ao de Aperfei\c{c}oamento
de Pessoal de N\'ivel Superior - Brasil (CAPES) -
Finance Code 001. Rafael Izbicki is
grateful for the financial support of
FAPESP (2017/03363-8) and CNPq (306943/2017-4).

\bibliography{main}

\appendix
\section{Appendix}

\subsection{Proofs}

\begin{proof}[Proof of \cref{thm:region-fcp}]
 \begin{align*}
  \gamma_{\phi}
  &= \int_{\tz \in H_0}\P_{\tz}(\phi(\X) = 1)f_{\theta}(\tz)d\tz
  +\int_{\to \in H_1}\P_{\to}(\phi(\X) = 0)f_{\theta}(\to)d\to 
  & \text{\cref{defn:level}} \\
  &= \int_{\tz \in H_0}
  \P_{\tz}(R(\X) \subset H_0^c)f_{\theta}(\tz)d\tz
  +\int_{\to \in H_1}
  \P_{\to}(R(\X) \subset H_0)f_{\theta}(\to)d\to
  & \text{\cref{defn:agnostic-region}} \\
  &\leq \int_{\tz \in H_0}
  \P_{\tz}(\tz \notin R(\X))f_{\theta}(\tz)d\tz
  +\int_{\to \in H_1}
  \P_{\to}(\to \notin R(\X))f_{\theta}(\to)d\to \\
  &\leq \int_{\ts \in \Theta} 
  \P_{\ts}(\ts \notin R(\X))f_{\theta}(\ts)d\ts
  = \gamma
  & \text{$R(\X)$ has credibility $(1-\gamma)$}
 \end{align*}
\end{proof}

\begin{theorem}
 \label{thm:cd_2pop}
 Let $\Z = (Z_0, Z_1)$ and 
 $\theta = (\theta_0, \theta_1)$ be such that
 $f(z_0,z_1|\theta_0,\theta_1) 
 = f(z_0|\theta_0)f(z_1|\theta_1)$. Also,
 $H_0: \theta_0 = \theta_1$,
 $Y \sim \text{Bernoulli}(0.5)$ and
 $Z^* = Z_{Y}$.
 If $\theta^* \notin Pg(H_0)$ and
 $\Z \sim f_{\theta^*}$, then
 it is possible to build
 a classifier with accuracy of
 at least $1-\epsilon$ for $Y$
 using $Z^*$.
\end{theorem}

\begin{proof}
 If $\ts \notin Pg(H_0)$, then
 for every $\bar{\theta} \in H_0$,
 $CD(\bar{\theta},\ts) \geq 1-\epsilon$.
 In particular, by choosing
 $\bar{\theta} = (\tso, \tso)$, obtain that
 $CD((\tso,\tso), \ts) \geq 1-\epsilon$.
 Therefore,
 \begin{align}
  \label{eq:2pop}
  1 - \epsilon 
  &\leq CD((\tso,\tso), \ts) \nonumber \\
  &= 0.5\left(\P_{(\tso,\tso)}\left(\frac{f_{(\tso,\tso)}(\Z)}{f_{\ts}(\Z)} > 1 \right)
 + \P_{\ts}\left(\frac{f_{\ts}(\Z)}{f_{(\tso,\tso)}(\Z)} > 1 \right) \right)
  & \text{\cref{defn:dissim}} \nonumber \\
  &= 0.5\left(\P_{(\tso,\tso)}\left(\frac{f_{\tso}(Z_1)f_{\tso}(Z_2)}{f_{\tso}(Z_1)f_{\tst}(Z_2)} > 1\right)
 + \P_{\ts}\left(\frac{f_{\tso}(Z_1)f_{\tst}(Z_2)}{f_{\tso}(Z_1)f_{\tso}(Z_2)} > 1 \right) \right)
 \nonumber \\
  &= 0.5\left(\P_{\tso}\left(\frac{f_{\tso}(Z_2)}{f_{\tst}(Z_2)} > 1 \right)
 + \P_{\tst}\left(\frac{f_{\tst}(Z_2)}{f_{\tso}(Z_2)} > 1 \right) \right)
 \end{align}
 The proof follows since \cref{eq:2pop} is
 the accuracy of the Bayes classifier
 \citep{Wasserman2013} for $Y$ using $Z^*$.
\end{proof}

\subsection{Approximate pragmatic hypotheses}
\label{sec:approx-pragm}

\begin{algorithm}
\caption{ \small Approximate pragmatic hypothesis computation for $H_0:\theta(1)=\theta_0$}\label{alg:pragmatic}
  \algorithmicrequire \ {\small  
  Null hypothesis parameter value $\theta_0$;
  estimates of $\theta(2),\ldots,\theta(d)$ based on
  the observed data, 
  $\widehat{\theta}(2),\ldots,\widehat{\theta}(d)$;
  dissimilarity threshold $0.5 \leq \epsilon \leq 1$;
  function log\_f($\z;\theta$) that computes the log-likelihood function of the new experiment;
  function generate\_samples($\theta$) that generates
  new samples $\Z$ from the distribution $f(\z|\theta)$; 
  number of Monte Carlo simulations $B$}
  \\
  \algorithmicensure \ {\small 
  Approximate pragmatic hypothesis
  $Pg^*(H_0)$}
  \begin{algorithmic}[1]
    \State Let $\bar{\theta}_0 \gets (\theta_0,\widehat{\theta}(2),\ldots,\widehat{\theta}(d))$
    \For{$i = 1,\ldots,B$}  
    \State  $\z^0_i = $ generate\_samples($\bar{\theta}_0$)
    \EndFor
    \State Let $Pg^*(H_0) \gets \emptyset$
    \For{$\theta^* \in \mathbb{R}$} 
    \State Let $\bar{\theta}^* \gets (\theta^*,\widehat{\theta}(2),\ldots,\widehat{\theta}(d))$
    \For{$i = 1,\ldots,B$}  
    \State  $\z^*_i = $ generate\_samples($\bar{\theta}^*$)
    \EndFor
    \State Let $\text{correct}_{\theta_0}\gets
    \text{mean}\left(\I\left(\text{log\_f}(\z^0_i;\bar{\theta}_0)>\text{log\_f}(\z^0_i;\bar{\theta}^*)\right)_{i=1}^B\right)$
    \State Let $\text{correct}_{\theta^*}\gets
    \text{mean}\left(\I\left(\text{log\_f}(\z^*_i;\bar{\theta}^*)>\text{log\_f}(\z^*_i;\bar{\theta}_0)\right)_{i=1}^B\right)$
    \State Let $\text{dist} = 2^{-1}\left(\text{correct}_{\theta_0} + \text{correct}_{\theta^*}\right)$
    \If{dist$<\epsilon$}
    \State $Pg^*(H_0) \gets Pg^*(H_0) \cup \{\theta^*\}$
    \EndIf
    \EndFor
    \State \textbf{return} $Pg^*(H_0)$
  \end{algorithmic}
\end{algorithm}

\end{document}